\RequirePackage{fixltx2e}
\documentclass[11pt,english,aps,manuscript]{revtex4}
\usepackage{mathpazo}

\usepackage[T1]{fontenc}
\usepackage[latin9]{inputenc}
\usepackage{color}
\usepackage{babel}
\usepackage{bm}
\usepackage{amsmath}
\usepackage{amsthm}
\usepackage{amssymb}
\usepackage[unicode=true,pdfusetitle,
 bookmarks=true,bookmarksnumbered=false,bookmarksopen=false,
 breaklinks=true,pdfborder={0 0 0},pdfborderstyle={},backref=false,colorlinks=true]
 {hyperref}
\hypersetup{
 urlcolor=blue, citecolor=blue, linkcolor=blue}

\makeatletter
\@ifundefined{textcolor}{}
{%
 \definecolor{BLACK}{gray}{0}
 \definecolor{WHITE}{gray}{1}
 \definecolor{RED}{rgb}{1,0,0}
 \definecolor{GREEN}{rgb}{0,1,0}
 \definecolor{BLUE}{rgb}{0,0,1}
 \definecolor{CYAN}{cmyk}{1,0,0,0}
 \definecolor{MAGENTA}{cmyk}{0,1,0,0}
 \definecolor{YELLOW}{cmyk}{0,0,1,0}
}
\IfFileExists{candleshape.def}{%
 \input{candleshape.def}}{}
\IfFileExists{dropshape.def}{%
 \input{dropshape.def}}{}
\IfFileExists{TeXshape.def}{%
 \input{TeXshape.def}}{}
\IfFileExists{triangleshapes.def}{%
 \input{triangleshapes.def}}{}

\theoremstyle{plain}
\newtheorem{thm}{\protect\theoremname}
\theoremstyle{remark}
\newtheorem{rem}{\protect\remarkname}
\theoremstyle{plain}
\newtheorem{lem}{\protect\lemmaname}
\theoremstyle{plain}
\newtheorem{cor}{\protect\corollaryname}
\theoremstyle{remark}
\newtheorem*{acknowledgement*}{\protect\acknowledgementname}

\makeatother

\providecommand{\acknowledgementname}{Acknowledgement}
\providecommand{\corollaryname}{Corollary}
\providecommand{\lemmaname}{Lemma}
\providecommand{\remarkname}{Remark}
\providecommand{\theoremname}{Theorem}

\begin{document}

\title{{\Large{}Locally distinguishing a maximally entangled basis using
shared entanglement }}

\author{Somshubhro Bandyopadhyay}

\affiliation{Department of Physical Sciences, Bose Institute, Bidhannagar, Kolkata
700091, India }
\email{som@jcbose.ac.in, som.s.bandyopadhyay@gmail.com}

\selectlanguage{english}%

\author{Vincent Russo}

\affiliation{Unitary Fund, 315 Montogomery St, Fl 10 San Francisco, California
94104, USA}
\email{vincent@unitary.fund}

\selectlanguage{english}%
\begin{abstract}
We consider the problem of distinguishing between the elements of
a bipartite maximally entangled orthonormal basis using local operations
and classical communication (LOCC) and a partially entangled state
acting as a resource. We derive an exact formula for the optimum success
probability and find that it corresponds to the fully entangled fraction
of the resource state. The derivation consists of two steps: First,
we consider a relaxation of the problem by replacing LOCC with positive-partial-transpose
(PPT) measurements and establish an upper bound on the success probability
as the solution of a semidefinite program, and then show that this
upper bound is achieved by a teleportation-based LOCC protocol. This
further implies that separable and PPT measurements provide no advantage
over LOCC for this task. We also present lower and upper bounds on
the success probability for distinguishing the elements of an incomplete
orthonormal maximally entangled basis in the same setup. 
\end{abstract}
\maketitle

\section{Introduction }

Distinguishing between quantum states is a fundamental problem in
quantum theory, where the aim is to ascertain the state of a quantum
system promised to be in one of a known set of states (see \citep{Chefles-review-2000,BC-review-2009,Bergou-2010,Bae-Kwek-review-2015}
for excellent reviews). The problem may be understood as follows.
Let $S_{\psi}=\left\{ \left(p_{i},\left|\psi_{i}\right\rangle \right):i=1,\dots,n\right\} $
be a set of quantum states, each state occurring with probability
$p_{i}$ such that $\sum_{i=1}^{n}p_{i}=1$. Now consider a quantum
system prepared in a state chosen from $S_{\psi}$, but we do not
know which one. The objective is to determine which state the system
is in by performing a suitable measurement. By measurement we mean
a positive operator valued measure (POVM) described by a collection
of positive operators satisfying the completeness relation. Now, according
to quantum theory, the given states can be perfectly distinguished
if and only if they are mutually orthogonal, and when they are not,
the best one can do is to perform a measurement which is optimal according
to some well-defined distinguishability measure. One such measure
is the success probability, the maximum probability that the unknown
state is identified correctly and is defined as
\begin{alignat}{1}
p\left(S_{\psi}\right) & =\sup_{\mathbf{M}}\sum_{i=1}^{n}p_{i}\left\langle \psi_{i}\right|M_{i}\left|\psi_{i}\right\rangle ,\label{success-prob}
\end{alignat}
where $\mathbf{M}=\left\{ M_{1},\dots,M_{n}\right\} $ is a measurement
and the supremum is taken over all measurements. Computing the success
probability in general is hard, and exact results have been found
only for specific instances of the problem (see \citep{Chefles-review-2000,BC-review-2009,Bae-Kwek-review-2015}
for comprehensive discussions). 

\subsection*{Local distinguishability}

This problem has also been studied within the \textquotedblleft distant
lab\textquotedblright{} paradigm in quantum information theory \citep{Peres-Wootters-1991,ben99,ben99u,Walgate-2000,Virmani-2001,Ghosh-2001,walgate-2002,HSSH-2003,divin03,Ghosh-2004,fan-2005,Watrous-2005,Nathanson-2005,Wootters-2006,Hayashi-etal-2006,nis06,Duan2007,feng09,Duan-2009,Calsamiglia-2010,BGK-2011,Bandyo-2011,Yu-Duan-2012,Cosentino-2013,Cosentino-Russo-2014,BN-2013,halder-2018,B-IQC-2015,Halder+-2019,Yu-duan-PPT-2014}.
For simplicity, we confine our discussion to bipartite quantum systems
of finite dimension $\mathbb{C}^{d}\otimes\mathbb{C}^{d}$, where
$d\geqslant2$. The given state now belongs to a known set of bipartite
states $S_{\phi}=\left\{ \left(p_{i},\left|\phi_{i}\right\rangle \right):i=1,\dots,n\right\} $
and is shared between two distant observers, Alice and Bob. The goal
is the same as before, but the measurements are restricted to \emph{local
operations and classical communication} (LOCC) \citep{LOCC}, wherein
Alice and Bob perform measurements on their respective local systems
and communicate via classical channels but cannot exchange quantum
information. The primary question of interest is how well LOCC can
distinguish the given states relative to global measurement (measurement
performed on the entire system). For example, if the states $\left|\phi_{i}\right\rangle $
are orthogonal, one would be interested in knowing whether they can
be perfectly distinguished by LOCC, as is always possible by a global
measurement. One motivation for studying this problem is to find out
how much global information encoded in orthogonal states of composite
systems is accessible by local means. Another motivation is to explore
nonlocal properties that may manifest in this setup. 

It turns out that in some instances of the problem, LOCC is optimal,
and in some, it is not. For example, two orthogonal pure states can
be perfectly distinguished by LOCC\footnote{For two orthogonal mixed states, surprisingly, this is not always
true \citep{Bandyo-2011,Yu-duan-PPT-2014}.} irrespective of the number of parties, dimensions, or entanglement
\citep{Walgate-2000}. However, this does not always hold for sets
of three or more orthogonal states. Notable examples of sets whose
members are not perfectly distinguishable by LOCC include the Bell
basis \citep{Ghosh-2001}, orthogonal bases containing entangled states
\citep{HSSH-2003}, unextendible product bases \citep{ben99u,divin03},
orthogonal product bases exhibiting nonlocality without entanglement
\citep{ben99,nis06,feng09,halder-2018} and strong nonlocality without
entanglement \citep{Halder+-2019}. Orthogonal states, which are perfectly
distinguishable by LOCC, are called locally distinguishable. Otherwise,
they are locally indistinguishable. Locally indistinguishable states,
product or otherwise, are said to be nonlocal in the sense a global
measurement can extract more information about the unknown state than
coordinated local measurements alone. These states have also found
novel applications in quantum cryptography primitives, such as secret
sharing \citep{Markham-Sanders-2008} and data hiding \citep{Terhal2001,DiVincenzo2002,Eggeling2002,MatthewsWehnerWinter}. 

Local (LOCC) distinguishability can be similarly quantified by (local)
success probability. Following \eqref{success-prob}, it is defined
as (see also \citep{Navascues-2008})
\begin{alignat}{1}
p_{_{{\rm L}}}\left(S_{\phi}\right) & =\sup_{\bm{\pi}\in{\rm LOCC}}\sum_{i=1}^{n}p_{i}\left\langle \phi_{i}\right|\pi_{i}\left|\phi_{i}\right\rangle ,\label{p-local-3}
\end{alignat}
where $\bm{\pi}=\left\{ \pi_{1},\dots,\pi_{n}\right\} $ is an LOCC
measurement. Computing the local optimum is notoriously hard even
for orthogonal states, for the LOCC class does not admit tractable
characterization. Nevertheless, a few results have been found. For
the members of the Bell basis, given with probabilities $p_{1}\geqslant p_{2}\geqslant p_{3}\geqslant p_{4}>0$,
the success probability is $p_{1}+p_{2}$ \citep{BN-2013,Navascues-2008}
and for any three uniformly distributed Bell states, it is $2/3$
\citep{B-IQC-2015}. In dimensions $\mathbb{C}^{d}\otimes\mathbb{C}^{d}$
for $d\geqslant3$, it has been shown that a set of $n$ maximally
entangled states can be locally distinguished with success probability
at most $d/n$ \citep{Nathanson-2005}, from which it follows that
no more than $d$ maximally entangled states can be perfectly distinguished
by LOCC; however, in some state spaces for $d\geqslant4$, one can
find even smaller locally indistinguishable sets for which $n\leqslant d$
\citep{Yu-Duan-2012,Cosentino-2013,Cosentino-Russo-2014}. 

\subsection*{Local distinguishability with shared entanglement}

By definition, LOCC is a strict subset of all quantum operations that
one may perform on a composite system; for example, LOCC can neither
create entanglement nor increase entanglement on average. The limitations
of LOCC, however, can be overcome with shared entanglement, which
is why entanglement is considered a resource for quantum information
processing tasks \citep{Entanglement-horodecki}. Naturally, one would
like to know how entanglement can help distinguish locally indistinguishable
states by LOCC \citep{B-IQC-2015,Cohen-2008,BGKW-2009,BRW-2010,BHN-2016,BHN-2018,LJ-2020,BR-2021}. 

One of the first questions addressed in this context deals with entanglement
cost, which quantifies how much entanglement must be consumed to perfectly
distinguish locally indistinguishable states by LOCC \citep{B-IQC-2015,Cohen-2008,BGKW-2009,BRW-2010,BHN-2018,BR-2021}.
It has been shown that, for example, a $\mathbb{C}^{2}\otimes\mathbb{C}^{2}$
maximally entangled state is necessary and sufficient to perfectly
distinguish the members of the Bell basis by LOCC \citep{B-IQC-2015}.
The LOCC protocol here mimics quantum teleportation \citep{teleportation}
and is known as the \textquotedblleft teleportation protocol,'' where
one-half of the unknown state is first teleported using the resource
state, followed by a Bell measurement. Likewise, a $\mathbb{C}^{d}\otimes\mathbb{C}^{d}$
maximally entangled state is necessary and sufficient to perfectly
distinguish a $\mathbb{C}^{d}\otimes\mathbb{C}^{d}$ maximally entangled
basis for all $d\geqslant3$ by LOCC. The proof of this fact follows
from a simple application of a result proved in \citep{HSSH-2003}.
Exact results, however, are hard to obtain for generic sets of orthonormal
states, even if they form a basis. Nevertheless, lower bounds have
been found; for example, a lower bound on the entanglement cost of
locally distinguishing an arbitrary orthonormal basis, not necessarily
maximally entangled, is given by the average entanglement of the basis
states \citep{BGKW-2009}, assuming the states are all equally probable.
In $\mathbb{C}^{2}\otimes\mathbb{C}^{2}$, this lower bound can be
improved upon for almost all orthonormal bases \citep{BRW-2010}. 

Instead of asking about entanglement cost, one could ask how well
locally indistinguishable states can be distinguished by LOCC using
shared entanglement as a resource \citep{B-IQC-2015,BHN-2016,BHN-2018,LJ-2020,BR-2021}.
The present work considers a problem of this kind, so let us discuss
the essentials of this problem in some detail. 

Let $\mathcal{A}_{1}=\mathcal{A}_{2}=\mathbb{C}^{d}$ and $\mathcal{B}_{1}=\mathcal{B}_{2}=\mathbb{C}^{d}$
denote the state spaces corresponding to the quantum systems held
by Alice and Bob, respectively. Let
\begin{alignat}{1}
S_{\chi} & =\left\{ \left(p_{i},\left|\chi_{i}\right\rangle \right):i=1,\dots,n\right\} \label{S-chi}
\end{alignat}
be an orthonormal set of locally indistinguishable states, where $\left|\chi_{i}\right\rangle \in\mathcal{A}_{1}\otimes\mathcal{B}_{1}$.
Suppose that Alice and Bob wish to distinguish the elements of $S_{\chi}$
using LOCC and a resource state $\left|\varphi\right\rangle \in\mathcal{A}_{2}\otimes\mathcal{B}_{2}$.
It is easy to see that this boils down to distinguishing between the
elements of the set 
\begin{alignat}{1}
S_{\chi\otimes\varphi} & =\left\{ \left(p_{i},\left|\chi_{i}\right\rangle \otimes\left|\varphi\right\rangle \right):i=1,\dots,n\right\} \label{eS(phiXvarphi)}
\end{alignat}
by LOCC, where $\left|\chi_{i}\right\rangle \otimes\left|\varphi\right\rangle \in\mathcal{A}_{1}\otimes\mathcal{B}_{1}\otimes\mathcal{A}_{2}\otimes\mathcal{B}_{2}$. 

Since Alice holds $\mathcal{A}_{1}$ and $\mathcal{A}_{2}$ and Bob
holds $\mathcal{B}_{1}$ and $\mathcal{B}_{2}$, LOCC is defined with
respect to the bipartition $\mathcal{A}:\mathcal{B}$, where $\mathcal{A}=\mathcal{A}_{1}\otimes\mathcal{A}_{2}$
and $\mathcal{B}=\mathcal{B}_{1}\otimes\mathcal{B}_{2}$. This reflects
the fact that Alice and Bob may perform joint measurements on the
composite systems $\mathcal{A}$ and $\mathcal{B}$, respectively.
So to define the success probability we need to take this into consideration,
which requires expressing the states $\left|\chi_{i}\right\rangle \otimes\left|\varphi\right\rangle $
as states in $\mathcal{A}\otimes\mathcal{B}$ by swapping the systems
$\mathcal{B}_{1}$ and $\mathcal{A}_{2}$. 

Define the unitary swap operator $U_{\mathcal{B}_{1}\leftrightarrow\mathcal{A}_{2}}:\mathcal{A}_{1}\otimes\mathcal{B}_{1}\otimes\mathcal{A}_{2}\otimes\mathcal{B}_{2}\rightarrow\mathcal{A}\otimes\mathcal{B}$
whose action on product states is given by
\begin{alignat}{1}
U_{\mathcal{B}_{1}\leftrightarrow\mathcal{A}_{2}}\left(\left|\alpha_{1}\right\rangle \left|\beta_{1}\right\rangle \left|\alpha_{2}\right\rangle \left|\beta_{2}\right\rangle \right) & =\left|\alpha_{1}\right\rangle \left|\alpha_{2}\right\rangle \left|\beta_{1}\right\rangle \left|\beta_{2}\right\rangle \label{eU-swap}
\end{alignat}
 for all vectors $\left|\alpha_{1}\right\rangle \in\mathcal{A}_{1}$,
$\left|\alpha_{2}\right\rangle \in\mathcal{A}_{2}$ , $\left|\beta_{1}\right\rangle \in\mathcal{B}_{1}$,
$\left|\beta_{2}\right\rangle \in\mathcal{B}_{2}$. 

Let $\left|\chi_{i}\right\rangle \otimes\left|\varphi\right\rangle \rightarrow\left|\xi_{i}\right\rangle \in\mathcal{A}\otimes\mathcal{B}$
for each $i=1,\dots,n$, where $\left|\xi_{i}\right\rangle =U_{\mathcal{B}_{1}\leftrightarrow\mathcal{A}_{2}}\left(\left|\chi_{i}\right\rangle \otimes\left|\varphi\right\rangle \right)$.
Accordingly, the success probability is defined as 
\begin{alignat}{1}
p_{_{{\rm L}}}\left(S_{\chi\otimes\varphi}\right) & =\sup_{\bm{\Pi}\in{\rm LOCC}}\sum_{i=1}^{n}p_{i}\left\langle \xi_{i}\right|\Pi_{i}\left|\xi_{i}\right\rangle ,\label{p-local-1}
\end{alignat}
where $\bm{\Pi}=\left\{ \Pi_{1},\dots,\Pi_{n}\right\} $ is a LOCC
measurement realized in the bipartition $\mathcal{A}:\mathcal{B}$. 

The success probability, defined in \eqref{p-local-1}, has been exactly
computed for sets of Bell states and a resource state of the form
$\left|\eta\right\rangle =a_{1}\left|00\right\rangle +a_{2}\left|11\right\rangle $,
where $a_{1}\geqslant a_{2}\geqslant0$ and $a_{1}^{2}+a_{2}^{2}=1$.
For the Bell basis, denoted by $S_{B}$, assuming uniform distribution,
it is given by \citep{B-IQC-2015}
\begin{alignat}{1}
p_{_{{\rm L}}}\left(S_{B\otimes\eta}\right) & =\frac{1}{2}\left(1+2a_{1}a_{2}\right).\label{p-Bell}
\end{alignat}
The success probability is unity if and only if $a_{1}=a_{2}=1/\sqrt{2}$,
so the Bell basis can be perfectly distinguished if and only if the
resource state is maximally entangled. For a set of any three uniformly
distributed Bell states, denoted by $S_{B^{\prime}}$, the success
probability turns out to be \citep{B-IQC-2015}
\begin{alignat}{1}
p_{_{{\rm L}}}\left(S_{B^{\prime}\otimes\eta}\right) & =\frac{2}{3}\left(1+a_{1}a_{2}\right).\label{p-3-Bell}
\end{alignat}
Observe that one still requires a maximally entangled state as a resource
to perfectly distinguish three Bell states by LOCC. In some sense,
this is counter-intuitive because one would have expected the entanglement
cost in this case to be less than that for the Bell basis. In addition
to these two examples, recently, the local optimum with shared entanglement
has been computed for a family of noisy Bell states, which shows,
yet again, the optimum value (i.e., the global optimum) is achieved
with a maximally entangled resource \citep{BR-2021}. 

The formulas in \eqref{p-Bell} and \eqref{p-3-Bell} are not just
ordinary functions of the Schmidt coefficients $a_{1}$ and $a_{2}$.
In particular, the one in \eqref{p-Bell} is the fully entangled fraction
of the resource state $\left|\eta\right\rangle $, defined as \citep{H-H-1999}
\begin{alignat}{1}
F\left(\eta\right) & =\max_{\left|\Psi\right\rangle }\left\langle \Psi\right|\eta\left|\Psi\right\rangle ,\label{FEF(eta)}
\end{alignat}
where $\eta=\left|\eta\right\rangle \left\langle \eta\right|$, and
the maximum is taken over all maximally entangled states $\left|\Psi\right\rangle $.
So we can write \eqref{p-Bell} simply as
\begin{alignat}{1}
p_{_{{\rm L}}}\left(S_{B\otimes\eta}\right) & =F\left(\eta\right).\label{p-FEF}
\end{alignat}
Therefore, how well the Bell basis can be distinguished using LOCC
and shared entanglement as a resource is given by how close the resource
state is to a maximally entangled state. Likewise, one may also express
\eqref{p-3-Bell} as a function of $F\left(\eta\right)$. 

Note that one may also express \eqref{p-Bell} and \eqref{p-3-Bell}
as a simple function of the entanglement of $\left|\eta\right\rangle $
quantified by an entanglement measure such as negativity $N\left(\eta\right)=a_{1}a_{2}$
\citep{VW-2002}.

\subsection*{Problem statement}

In the present work, we consider the problem of locally distinguishing
an orthonormal maximally entangled basis of $\mathcal{A}_{1}\otimes\mathcal{B}_{1}$
using a resource state in $\mathcal{A}_{2}\otimes\mathcal{B}_{2}$.
As noted earlier, the resource state must be maximally entangled to
perfectly distinguish the elements of such a basis. So, here we are
interested in knowing how well such a basis can be locally distinguished
using a partially entangled resource. This can be answered by computing
the optimum local success probability. In addition, we will also discuss
the case when the maximally entangled states do not form a complete
basis. 

One may be tempted to view the primary problem as a simple extension
of the Bell basis problem previously discussed. However, we want to
emphasize the notable differences. First is that the Bell basis is
a particular basis of $\mathbb{C}^{2}\otimes\mathbb{C}^{2}$, whereas,
in this paper, we do not assume any particular basis, such as the
canonical generalized Bell basis used in quantum teleportation, and
the analysis holds for any orthonormal maximally entangled basis of
$\mathbb{C}^{d}\otimes\mathbb{C}^{d}$. The second is that for the
Bell basis, though the optimum success probability is given by the
fully entangled fraction of the resource state, there is no particular
reason to assume that this will carry over to higher dimensions. In
fact, results from local distinguishability tells us that many results
that hold in lower dimensions fail to hold in higher dimensions; for
example, any two maximally entangled states in $\mathbb{C}^{2}\otimes\mathbb{C}^{2}$
or any three maximally entangled states in $\mathbb{C}^{3}\otimes\mathbb{C}^{3}$
are perfectly distinguishable by LOCC but in dimensions $\mathbb{C}^{d}\otimes\mathbb{C}^{d}$
for $d\geqslant4$, a similar result does not hold \citep{Yu-Duan-2012,Cosentino-2013,Cosentino-Russo-2014}.
While we expect the optimal success probability (assuming an exact
formula can be obtained) to be a function of the Schmidt coefficients
of the resource state, the same form cannot be predicted beforehand.
The third is that for the Bell basis, the calculations were greatly
simplified because the states were known and enjoy a lot of symmetry,
whereas, in our case, since we do not assume any particular basis,
such simplifications are not possible. 

Let $\left\{ \left|\Psi_{1}\right\rangle ,\dots,\left|\Psi_{d^{2}}\right\rangle \right\} $
be an orthonormal, maximally entangled basis of $\mathcal{A}_{1}\otimes\mathcal{B}_{1}$.
Without loss of generality, we assume that
\begin{alignat}{1}
\left|\Psi_{1}\right\rangle  & =\frac{1}{\sqrt{d}}\sum_{i=1}^{d}\left|i\right\rangle \left|i\right\rangle \label{Psi(1)}
\end{alignat}
is the standard maximally entangled state. The remaining states can
be written as
\begin{alignat}{1}
\left|\Psi_{j}\right\rangle  & =\left(\mathbf{1}_{\mathcal{A}_{1}}\otimes U_{j}\right)\left|\Psi_{1}\right\rangle \label{Psi(k)}
\end{alignat}
for some unitary operator $U_{j}\in U\left(\mathcal{B}_{1}\right)$
for each $j=2,\dots,d^{2}$, where $U\left(\mathcal{B}_{1}\right)$
denotes the set of unitary operators acting on $\mathcal{B}_{1}$.
Since $\left\langle \Psi_{i}\vert\Psi_{j}\right\rangle =\delta_{ij}$
for all $i,j=1,\dots,d^{2}$, the unitary operators obey the relation
$\text{Tr}\left(U_{i}^{\dagger}U_{j}\right)=d\delta_{ij}$ for $i,j=1,\dots,d^{2}$,
where $U_{1}=\bm{1}_{\mathcal{B}_{1}}$. 

Define the set
\begin{alignat}{1}
S_{\Psi} & =\left\{ \left(\frac{1}{d^{2}},\left|\Psi_{k}\right\rangle \right):k=1,\dots,d^{2}\right\} .\label{S(Psi)}
\end{alignat}
This set is locally indistinguishable \citep{HSSH-2003,Nathanson-2005}
with the local optimum given by $p_{_{{\rm L}}}\left(S_{\Psi}\right)=1/d$
\citep{Nathanson-2005,Cosentino-2013}. 

In this paper, we consider the problem of distinguishing between the
elements of $S_{\Psi}$ using LOCC and a resource state 
\begin{alignat}{1}
\left|\tau\right\rangle  & =\sum_{i=1}^{d}a_{i}\left|i\right\rangle \left|i\right\rangle \in\mathcal{A}_{2}\otimes\mathcal{B}_{2},\label{tau}
\end{alignat}
where $\left\{ a_{i}\right\} $ are the ordered Schmidt coefficients
($a_{1}\geqslant a_{2}\geqslant\cdots\geqslant a_{d}\geqslant0$)
satisfying $\sum_{i=1}^{d}a_{i}^{2}=1$. Note that the resource state
is entangled as long as two or more Schmidt coefficients are positive. 

The set of interest is, therefore, 
\begin{alignat}{1}
S_{\Psi\otimes\tau} & =\left\{ \left(\frac{1}{d^{2}},\left|\Psi_{k}\right\rangle \otimes\left|\tau\right\rangle \right):k=1,\dots,d^{2}\right\} .\label{psicrosstau}
\end{alignat}
We wish to find out how well the elements of $S_{\Psi\otimes\tau}$
can be distinguished by LOCC. Since this problem is defined in the
bipartition $\mathcal{A}:\mathcal{B}$ with $\mathcal{A}=\mathcal{A}_{1}\otimes\mathcal{A}_{2}$
and $\mathcal{B}=\mathcal{B}_{1}\otimes\mathcal{B}_{2}$, we first
let $\left|\Psi_{k}\right\rangle \otimes\left|\tau\right\rangle \rightarrow\left|\Phi_{k}\right\rangle \in\mathcal{A}\otimes\mathcal{B}$,
where $\left|\Phi_{k}\right\rangle =U_{\mathcal{B}_{1}\leftrightarrow\mathcal{A}_{2}}\left(\left|\Psi_{k}\right\rangle \otimes\left|\tau\right\rangle \right)$
for each $k=1,\dots,d^{2}$, and define the success probability as
\begin{alignat}{1}
p_{_{{\rm L}}}\left(S_{\Psi\otimes\tau}\right) & =\sup_{\bm{\Pi}\in{\rm LOCC}}\frac{1}{d^{2}}\sum_{i=1}^{d^{2}}\left\langle \Phi_{i}\right|\Pi_{i}\left|\Phi_{i}\right\rangle ,\label{p-local-1-1}
\end{alignat}
where $\bm{\Pi}=\left\{ \Pi_{1},\dots,\Pi_{d^{2}}\right\} $ is a
LOCC measurement realized in the bipartition $\mathcal{A}:\mathcal{B}$. 

Our objective is to compute the local optimum defined in \eqref{p-local-1-1}
and find an LOCC protocol achieving the same value. As far as we are
aware, this problem remains open in all state spaces $\mathbb{C}^{d}\otimes\mathbb{C}^{d}$
for $d\geqslant3$. We are particularly interested in knowing whether
an expression similar to \eqref{p-FEF} holds for $p_{_{{\rm L}}}\left(S_{\Psi\otimes\tau}\right)$
as well.

\subsection*{Overview of the main results }

The main result of the paper is the following: 
\begin{thm}
\label{main-theorem} The success probability for distinguishing the
elements of $S_{\Psi\otimes\tau}$ by LOCC is given by 
\begin{alignat}{1}
p_{_{{\rm L}}}\left(S_{\Psi\otimes\tau}\right) & =F\left(\tau\right),\label{eqTheorem-1}
\end{alignat}
where $F\left(\tau\right)$ is the fully entangled fraction of $\left|\tau\right\rangle $,
defined as the maximum overlap between $\left|\tau\right\rangle $
and a maximally entangled state, and given by the formula \citep{H-H-1999}
\begin{alignat}{1}
F\left(\tau\right) & =\frac{1}{d}\left(\sum_{i=1}^{d}a_{i}\right)^{2}\in\left[\frac{1}{d},1\right].\label{FEF(tau)}
\end{alignat}
The optimum value is achieved by a teleportation-based LOCC protocol. 
\end{thm}
\begin{rem}
Since $\frac{1}{d}\leqslant F\left(\tau\right)\leqslant1$, it holds
that $\frac{1}{d}\leqslant p_{_{{\rm L}}}\left(S_{\Psi\otimes\tau}\right)\leqslant1$,
where each of the left and right inequalities become equality for
product and maximally entangled $\left|\tau\right\rangle $, respectively. 
\begin{rem}
Theorem \ref{main-theorem} holds for any maximally entangled basis
of $\mathbb{C}^{d}\otimes\mathbb{C}^{d}$. Our proof, as we will see,
is representation independent.
\begin{rem}
One may also express the success probability as a function of the
negativity of the resource state $N\left(\tau\right)$ using the relation
$F\left(\tau\right)=\frac{1}{d}\left[1+2N\left(\tau\right)\right]$. 
\end{rem}
\end{rem}
\end{rem}
To prove Theorem \ref{main-theorem} we will proceed as follows. First
we will consider a relaxation of our problem by replacing LOCC with
positive-partial-transpose (PPT) measurements (these are measurements
whose operators are positive under partial transposition). The (PPT)
success probability is defined as 
\begin{alignat}{1}
p_{_{{\rm PPT}}}\left(S_{\Psi\otimes\tau}\right) & =\sup_{\bm{\Omega}\in{\rm PPT}}\frac{1}{d^{2}}\sum_{i=1}^{d^{2}}\left\langle \Phi_{i}\right|\Omega_{i}\left|\Phi_{i}\right\rangle ,\label{p-local-1-1-1}
\end{alignat}
where $\bm{\Omega}=\left\{ \Omega_{1},\dots,\Omega_{d^{2}}\right\} $
is a PPT measurement on $\mathcal{A}\otimes\mathcal{B}$. 

Following \citep{Cosentino-2013}, we will formulate our PPT distinguishability
problem as a semidefinite program (SDP) and solve the dual problem
to obtain an upper bound on $p_{_{{\rm PPT}}}\left(S_{\Psi\otimes\tau}\right)$. 
\begin{lem}
\label{PPT-upper bound} An upper bound on the success probability
for distinguishing the elements of $S_{\Psi\otimes\tau}$ by any PPT
measurement is given by 
\begin{alignat}{1}
p_{_{{\rm PPT}}}\left(S_{\Psi\otimes\tau}\right) & \leqslant F\left(\tau\right).\label{PPT-upper bound-1}
\end{alignat}
\end{lem}
Since LOCC is a strict subset of PPT measurements, we have the following
corollary.
\begin{cor}
\label{p_L<=00003Dp(PPT)-1} For the given set of states $S_{\Psi\otimes\tau}$,
it holds that
\begin{alignat}{1}
p_{_{{\rm L}}}\left(S_{\Psi\otimes\tau}\right)\leqslant p_{{\rm _{{\rm PPT}}}}\left(S_{\Psi\otimes\tau}\right) & \leqslant F\left(\tau\right).\label{p(L)<=00003Dp(PPT)}
\end{alignat}
\end{cor}
The next result shows that the upper bound in Corollary \ref{p_L<=00003Dp(PPT)-1}
is also a lower bound on $p_{_{{\rm L}}}\left(S_{\Psi\otimes\tau}\right)$. 
\begin{lem}
\label{LOCC-prob} The success probability for distinguishing the
elements of $S_{\Psi\otimes\tau}$ by LOCC is bounded below by

\begin{alignat}{1}
p_{_{{\rm L}}}\left(S_{\Psi\otimes\tau}\right) & \geqslant F\left(\tau\right).\label{p(L)>=00003D()}
\end{alignat}
\end{lem}
We will prove Lemma \ref{LOCC-prob} by presenting an LOCC protocol
that distinguishes between the elements of $S_{\Psi\otimes\tau}$
with probability $F\left(\tau\right)$. The protocol here is based
on quantum teleportation. 

Since the lower bound in Lemma \ref{LOCC-prob} matches the upper
bound in Corollary \ref{p_L<=00003Dp(PPT)-1}, this completes the
proof of Theorem \ref{main-theorem}. Furthermore, 
\begin{alignat}{1}
p_{_{{\rm L}}}\left(S_{\Psi\otimes\tau}\right) & =p_{_{{\rm PPT}}}\left(S_{\Psi\otimes\tau}\right)=F\left(\tau\right),\label{p(L)=00003Dp(PPT)}
\end{alignat}
which shows that, even though $\text{LOCC}\subset\text{PPT}$, the
LOCC optimum equals the PPT optimum. Thus PPT measurements (hence,
separable measurements) provide no advantage in distinguishing a maximally
entangled basis using shared entanglement. 

The rest of the paper is arranged as follows. In Sec.$\,$\ref{PPT-SDP},
we discuss the SDP formulation of distinguishing a set of states by
PPT measurements, including ours. We prove Lemma \ref{PPT-upper bound}
in Sec.$\,$\ref{Lemma upper bound} and Lemma \ref{LOCC-prob} in
Sec.$\,$\ref{Lemma-lower-bound}. In Sec.$\,$\ref{incomplete} we
present lower and upper bounds on the success probability for distinguishing
the elements of an incomplete orthonormal maximally entangled basis
in the same setup. We conclude in Sec.$\,$\ref{Conclusions} with
a brief summary of results and a discussion on open problems. 

\section{\label{PPT-SDP} PPT distinguishability as a semidefinite program}

The problem of distinguishing a set of states by PPT measurements
can be cast as a semidefinite program \citep{Cosentino-2013} and,
thereby, can be solved for many problems of interest. This, coupled
with the fact that ${\rm LOCC}\subset{\rm PPT}$ has yielded exact
results \citep{Cosentino-2013,Cosentino-Russo-2014,B-IQC-2015,Yu-duan-PPT-2014,BR-2021},
no-go results \citep{Yu-Duan-2012,Yu-duan-PPT-2014}, and useful bounds
\citep{Cosentino-2013,Yu-duan-PPT-2014} for local (in)distinguishability
problems that were once thought to be intractable.

Let $\mathcal{X}$ and $\mathcal{Y}$ represent $d$-dimensional state
spaces for $d\geqslant2$. Let ${\rm Pos}\left(\mathcal{X}\right)$,
${\rm Pos}\left(\mathcal{Y}\right)$, and ${\rm Pos}\left(\mathcal{X}\otimes\mathcal{Y}\right)$
denote the sets of positive semidefinite operators acting on $\mathcal{X}$,
$\mathcal{Y}$, and $\mathcal{X}\otimes\mathcal{Y}$, respectively.
An operator $P\in{\rm Pos}\left(\mathcal{X}\otimes\mathcal{Y}\right)$
is PPT if ${\rm T}_{\mathcal{X}}\left(P\right)\in{\rm Pos}\left(\mathcal{X}\otimes\mathcal{Y}\right)$,
where $\text{T}_{\mathcal{X}}$ represents partial transposition taken
in the standard basis of $\mathcal{X}$ (note that, as far as the
definition of PPT is concerned, partial transposition could be taken
with respect to any one of the state spaces $\mathcal{X}$ or $\mathcal{Y}$). 

Denote the set of all PPT operators acting on $\mathcal{X}\otimes\mathcal{Y}$
by $\text{PPT}\left(\mathcal{X}:\mathcal{Y}\right)$. The set $\text{PPT}\left(\mathcal{X}:\mathcal{Y}\right)$
is a closed, convex cone. A PPT measurement is defined by a collection
of measurement operators $\left\{ P_{1},\dots,P_{n}\right\} $, where
$P_{i}\in{\rm PPT}\left(\mathcal{X}:\mathcal{Y}\right)$ for each
$i=1,\dots,n$. 

For a given ensemble $\mathcal{E}=\left\{ \left(p_{i},\rho_{i}\right):i=1,\dots,n\right\} $,
where $\rho_{i}$ are density operators on $\mathcal{X}\otimes\mathcal{Y}$,
the problem of finding $p_{_{{\rm PPT}}}\left(\mathcal{E}\right)$
can be expressed as a semidefinite program \citep{Cosentino-2013}:
\begin{center}
\begin{eqnarray*}
{\rm \text{Primal problem}} & \hspace{7em} & \text{Dual problem}\\
\mathtt{maximize}:\;\sum_{i=1}^{n}p_{i}\text{Tr}\left(\rho_{i}P_{i}\right) &  & \mathtt{minimize}:\;\text{Tr}\left(H\right)\\
\mathtt{subject\,\mathtt{to}}:\;\sum_{i=1}^{n}P_{i}=\mathbf{1}_{\mathcal{X}\otimes\mathcal{Y}} &  & \mathtt{subject}\,\mathtt{to}:\;H-p_{k}\rho_{k}\in\text{PPT}\left(\mathcal{X}:\mathcal{Y}\right)\\
P_{k}\in\text{PPT}\left(\mathcal{X}:\mathcal{Y}\right) &  & H\in\text{Herm}\left(\mathcal{X}\otimes\mathcal{Y}\right)\\
\left(k=1,\dots,n\right) &  & \left(k=1,\dots,n\right)
\end{eqnarray*}
\par\end{center}

where $\text{Herm}\left(\mathcal{X}\otimes\mathcal{Y}\right)$ is
the set of Hermitian operators acting on $\mathcal{X}\otimes\mathcal{Y}$. 

Let $\omega$ denote the solution of the dual problem. By the weak
duality theorem, it holds that $p_{_{{\rm PPT}}}\left(\mathcal{E}\right)\leqslant\omega$.
Thus, every feasible solution of the dual problem provides an upper
bound on $p_{_{{\rm PPT}}}\left(\mathcal{E}\right)$. 

\subsection*{Distinguishing the elements of $S_{\Psi\otimes\tau}$ by PPT measurements:
SDP formulation}

Following the above prescription, the primal and dual problems for
distinguishing the elements of $S_{\Psi\otimes\tau}$ by a PPT measurement
are the following:
\begin{center}
\begin{eqnarray*}
\text{Primal problem} & \hspace{7em} & \text{Dual problem}\\
\mathtt{maximize}:\;\frac{1}{d^{2}}\sum_{i=1}^{d^{2}}\left\langle \Phi_{i}\right|\Omega_{i}\left|\Phi_{i}\right\rangle  &  & \mathtt{minimize}:\;\text{Tr}\left(H\right)\\
\mathtt{subject\,\mathtt{to}}:\;\sum_{i=1}^{d^{2}}\Omega_{i}=\mathbf{1}_{\mathcal{A}\otimes\mathcal{B}} &  & \mathtt{subject}\,\mathtt{to}:\;H-\Phi_{k}/d^{2}\in\text{PPT}\left(\mathcal{A}:\mathcal{B}\right)\\
\Omega_{k}\in\text{PPT}\left(\mathcal{A}:\mathcal{B}\right) &  & H\in\text{Herm}\left(\mathcal{A}\otimes\mathcal{B}\right)\\
\left(k=1,\dots,d^{2}\right) &  & \left(k=1,\dots,d^{2}\right)
\end{eqnarray*}
\par\end{center}

where $\Phi_{k}=\left|\Phi_{k}\right\rangle \left\langle \Phi_{k}\right|$$\text{ and Herm}\left(\mathcal{A}\otimes\mathcal{B}\right)$
is the set of Hermitian operators acting on $\mathcal{A}\otimes\mathcal{B}$. 

Therefore, any $H\in\text{Herm}\left(\mathcal{A}\otimes\mathcal{B}\right)$
for which the constraint of the dual problem is satisfied, it holds
that
\begin{alignat}{1}
p_{_{{\rm PPT}}}\left(S_{\Psi\otimes\tau}\right) & \leqslant{\rm Tr}\left(H\right).\label{p<=00003DTr(H)}
\end{alignat}
 So if we could find an appropriate $H$, we would immediately obtain
an upper bound on the local optimum $p_{_{{\rm L}}}\left(S_{\Psi\otimes\tau}\right)$
as $p_{_{{\rm L}}}\left(S_{\Psi\otimes\tau}\right)\leqslant p_{_{{\rm PPT}}}\left(S_{\Psi\otimes\tau}\right)$.

\section{\label{Lemma upper bound} Upper bound on the success probability:
Proof of Lemma \ref{PPT-upper bound}}

Our objective is to find an $H\in{\rm Herm}\left(\mathcal{A}\otimes\mathcal{B}\right)$
that satisfies the dual feasibility condition 
\begin{alignat}{1}
{\rm T}_{\mathcal{A}}\left(H-\Phi_{k}/d^{2}\right) & ={\rm T}_{\mathcal{A}}\left(H\right)-\frac{1}{d^{2}}{\rm T}_{\mathcal{A}}\left(\Phi_{k}\right)\in{\rm Pos}\left(\mathcal{A}\otimes\mathcal{B}\right),\;k=1,\dots,d^{2},\label{dual constraint}
\end{alignat}
and for which ${\rm Tr}\left(H\right)=F\left(\tau\right)$. Once an
appropriate $H$ is found, the proof will then follow from \eqref{p<=00003DTr(H)}. 

Let $\mathbb{H}$ be a Hermitian operator of the form 
\begin{alignat}{1}
\mathbb{H} & =\sum_{i=1}^{m}h_{1i}\otimes h_{2i}\in{\rm Herm}\left(\mathcal{A}_{1}\otimes\mathcal{B}_{1}\otimes\mathcal{A}_{2}\otimes\mathcal{B}_{2}\right),\label{Hprime-1}
\end{alignat}
where $h_{1i}\in{\rm Herm}\left(\mathcal{A}_{1}\otimes\mathcal{B}_{1}\right)$
and $h_{2i}\in{\rm Herm}\left(\mathcal{A}_{2}\otimes\mathcal{B}_{2}\right)$
for each $i=1,.\dots,m$. We will use the following lemma to find
a feasible solution of the dual problem. 
\begin{lem}
\label{H and Hprime} Suppose $\mathbb{H}$ is a Hermitian operator
of the form \eqref{Hprime-1} such that 
\begin{alignat}{1}
{\rm \left(T_{\mathcal{A}_{1}}\otimes{\rm T}_{\mathcal{A}_{2}}\right)}\left(\mathbb{H}-\frac{1}{d^{2}}\Psi_{k}\otimes\tau\right) & \in{\rm Pos}\left(\mathcal{A}_{1}\otimes\mathcal{B}_{1}\otimes\mathcal{A}_{2}\otimes\mathcal{B}_{2}\right)\label{PositivityH}
\end{alignat}
for each $k=1,\dots,d^{2}$, where $\Psi_{k}=\left|\Psi_{k}\right\rangle \left\langle \Psi_{k}\right|$
and $\tau=\left|\tau\right\rangle \left\langle \tau\right|$. Then
the dual feasibility condition \eqref{dual constraint} is satisfied
for 
\begin{alignat}{1}
H & =U_{\mathcal{B}_{1}\leftrightarrow\mathcal{A}_{2}}\mathbb{H}U_{\mathcal{B}_{1}\leftrightarrow\mathcal{A}_{2}}^{\dagger}\in{\rm Herm\left(\mathcal{A}\otimes\mathcal{B}\right)}.\label{H=00003DH^=00005Cprime}
\end{alignat}
\end{lem}
The proof is given in the Appendix. \\
\\
First observe that the equality 
\begin{alignat}{1}
{\rm Tr}\left(H\right) & ={\rm Tr}\left(\mathbb{H}\right)\label{Tr(H)=00003DTr(Hprime)}
\end{alignat}
follows immediately from \eqref{H=00003DH^=00005Cprime}. 

Therefore, to prove Lemma \ref{PPT-upper bound} it is sufficient
to find an $\mathbb{H}$ of the form \eqref{Hprime-1} such that \eqref{PositivityH}
is satisfied for all $k=1,\dots,d^{2}$, and for which ${\rm Tr}\left(\mathbb{H}\right)=F\left(\tau\right)$. 

Define the operator 
\begin{alignat}{1}
\mathbb{H} & =\frac{1}{d^{3}}\mathbf{1}_{\mathcal{A}_{1}\otimes\mathcal{B}_{1}}\otimes\left[\tau+2\sum_{i,j=1,\,i<j}^{d}a_{i}a_{j}{\rm T}_{\mathcal{A}_{2}}\left(\psi_{ij}^{-}\right)\right]\in\text{Herm}\left(\mathcal{A}_{1}\otimes\mathcal{B}_{1}\otimes\mathcal{A}_{2}\otimes\mathcal{B}_{2}\right),\label{H}
\end{alignat}
where $\psi_{ij}^{-}=\left|\psi_{ij}^{-}\right\rangle \left\langle \psi_{ij}^{-}\right|$,
$\left|\psi_{ij}^{-}\right\rangle =\frac{1}{\sqrt{2}}\left(\left|i\right\rangle \left|j\right\rangle -\left|j\right\rangle \left|i\right\rangle \right)$,
and ${\rm T}_{\mathcal{A}_{2}}$ denotes the partial transposition
of with respect to the standard basis of $\mathcal{A}_{2}$. 

Observe that $\mathbb{H}$ has the form \eqref{Hprime-1}, and a simple
calculation shows 
\begin{alignat}{1}
\text{Tr}\left(\mathbb{H}\right) & =\frac{1}{d}\left(1+2\sum_{i,j=1,\,i<j}^{d}a_{i}a_{j}\right)=F\left(\tau\right),\label{Tr(H)}
\end{alignat}
where the second equality follows from \eqref{FEF(tau)}. We now prove
that the positivity condition \eqref{PositivityH} is satisfied. 

First, observe that
\begin{flalign}
\left({\rm T}_{\mathcal{A}_{1}}\otimes{\rm T}_{\mathcal{A}_{2}}\right)\left[\mathbb{H}-\frac{1}{d^{2}}\left(\Psi_{k}\otimes\tau\right)\right] & =\left({\rm T}_{\mathcal{A}_{1}}\otimes{\rm T}_{\mathcal{A}_{2}}\right)\mathbb{H}-\frac{1}{d^{2}}\left({\rm T}_{\mathcal{A}_{1}}\otimes{\rm T}_{\mathcal{A}_{2}}\right)\left(\Psi_{k}\otimes\tau\right)\nonumber \\
 & =\frac{1}{d^{3}}\mathbf{1}_{\mathcal{A}_{1}\otimes\mathcal{B}_{1}}\otimes\left[{\rm T}_{\mathcal{A}_{2}}\left(\tau\right)+2\sum_{i,j=1,\,i<j}^{d}a_{i}a_{j}\psi_{ij}^{-}\right]-\frac{1}{d^{2}}{\rm T}_{\mathcal{A}_{1}}\left(\Psi_{k}\right)\otimes{\rm T_{\mathcal{A}_{2}}}\left(\tau\right)\nonumber \\
 & =\frac{1}{d^{3}}\Upsilon_{k}\otimes{\rm T}_{\mathcal{A}_{2}}\left(\tau\right)+\mathbf{1}_{\mathcal{A}_{1}\otimes\mathcal{B}_{1}}\otimes\frac{2}{d^{3}}\sum_{i,j=1,\,i<j}^{d}a_{i}a_{j}\psi_{ij}^{-}.\label{Theta-=00005CLambda}
\end{flalign}
where
\begin{alignat}{1}
\Upsilon_{k} & =\mathbf{1}_{\mathcal{A}_{1}\otimes\mathcal{B}_{1}}-d{\rm T}_{\mathcal{A}_{1}}\left(\Psi_{k}\right).\label{upsilon-k}
\end{alignat}
 We will now evaluate the right-hand-side (RHS) of \eqref{Theta-=00005CLambda}. 

A straightforward calculation reveals

\begin{flalign}
{\rm T_{\mathcal{A}_{2}}}\left(\tau\right) & =\sum_{i=1}^{d}a_{i}^{2}\psi_{ii}+\sum_{i,j=1,\,i<j}^{d}a_{i}a_{j}\psi_{ij}^{+}-\sum_{i,j=1,\,i<j}^{d}a_{i}a_{j}\psi_{ij}^{-},\label{T(Tau)}
\end{flalign}
where $\psi_{ii}=\left|\psi_{ii}\right\rangle \left\langle \psi_{ii}\right|$,
$\left|\psi_{ii}\right\rangle =\left|i\right\rangle \left|i\right\rangle $
and $\psi_{ij}^{+}=\left|\psi_{ij}^{+}\right\rangle \left\langle \psi_{ij}^{+}\right|$,
$\left|\psi_{ij}^{+}\right\rangle =\frac{1}{\sqrt{2}}\left(\left|i\right\rangle \left|j\right\rangle +\left|j\right\rangle \left|i\right\rangle \right)$.
Let 
\begin{alignat}{1}
\Gamma & =\sum_{i=1}^{d}a_{i}^{2}\psi_{ii}+\sum_{i,j=1,\,i<j}^{d}a_{i}a_{j}\psi_{ij}^{+}\label{Gamma}
\end{alignat}
which is clearly positive semidefinite (since $a_{i}\geqslant0$ for
all $i=1,\dots,d$, and $\left\{ \psi_{ii}\right\} $ and $\left\{ \psi_{ij}^{+}\right\} $
are density operators) and write \eqref{T(Tau)} in a compact form
\begin{alignat}{1}
{\rm T_{\mathcal{A}_{2}}}\left(\tau\right) & =\Gamma-\sum_{i,j=1,\,i<j}^{d}a_{i}a_{j}\psi_{ij}^{-}.\label{T(Tau)-1}
\end{alignat}
Using \eqref{T(Tau)-1} in \eqref{Theta-=00005CLambda} and simplifying,
we get
\begin{alignat}{1}
\left({\rm T}_{\mathcal{A}_{1}}\otimes{\rm T}_{\mathcal{A}_{2}}\right)\left[\mathbb{H}-\frac{1}{d^{2}}\left(\Psi_{k}\otimes\tau\right)\right] & =\frac{1}{d^{3}}\Upsilon_{k}\otimes\Gamma+\frac{2}{d^{3}}\sum_{i,j=1,\,i<j}^{d}a_{i}a_{j}\left(\mathbf{1}_{\mathcal{A}_{1}\otimes\mathcal{B}_{1}}-\frac{1}{2}\Upsilon_{k}\right)\otimes\psi_{ij}^{-}.\label{Final term}
\end{alignat}
Since $\Gamma$ and the density operators $\left\{ \psi_{ij}^{-}\right\} $,
both being independent of $k$, are positive semidefinite, it suffices
to prove that $\Upsilon_{k}$ and $\left(\mathbf{1}_{\mathcal{A}_{1}\otimes\mathcal{B}_{1}}-\frac{1}{2}\Upsilon_{k}\right)$
are positive semidefinite for each $k=1,\dots,d^{2}$. 

To calculate $\Upsilon_{k}$, we need to first compute ${\rm T}_{\mathcal{A}_{1}}\left(\Psi_{k}\right)$.
This requires expressing $\left|\Psi_{k}\right\rangle $ in an orthonormal
product basis of the form $\left\{ \left|\mu_{i}\right\rangle \left|\nu_{j}\right\rangle \right\} _{i,j=1}^{d}$,
where $\left\{ \left|\mu_{i}\right\rangle \right\} _{i=1}^{d}$ and
$\left\{ \left|\nu_{j}\right\rangle \right\} _{j=1}^{d}$ are the
orthonormal bases of $\mathcal{A}_{1}$ and $\mathcal{B}_{1}$, respectively,
and then taking the partial transpose in the basis of $\mathcal{A}_{1}$.
Since positive semidefiniteness of the operators under consideration
does not depend on the choice of the product basis used for taking
the partial transposition, it will be convenient to take this product
basis as the Schmidt basis of $\left|\Psi_{k}\right\rangle $ for
each $k$ (note that the Schmidt basis, in general, will be different
for different $\left|\Psi_{k}\right\rangle $). We therefore write
$\left|\Psi_{k}\right\rangle $ in the Schmidt form
\begin{alignat}{1}
\left|\Psi_{k}\right\rangle  & =\frac{1}{\sqrt{d}}\sum_{i=1}^{d}\left|\alpha_{k}^{i}\right\rangle \left|\beta_{k}^{i}\right\rangle ,\label{Psi_k-Schmidt}
\end{alignat}
where $\left\{ \left|\alpha_{k}^{1}\right\rangle ,\dots,\left|\alpha_{k}^{d}\right\rangle \right\} $
and $\left\{ \left|\beta_{k}^{1}\right\rangle ,\dots,\left|\beta_{k}^{d}\right\rangle \right\} $
are orthonormal bases of $\mathcal{A}_{1}$ and $\mathcal{B}_{1}$,
respectively. It holds that 
\begin{alignat}{1}
{\rm {\rm T}}_{\mathcal{A}_{1}}\left(\Psi_{k}\right) & =\frac{1}{d}\left[\sum_{i}\xi_{k}^{ii}+\sum_{i,j=1,\,i<j}^{d}\xi_{k}^{ij^{+}}-\sum_{i,j=1,\,i<j}^{d}\xi_{k}^{ij^{-}}\right],\label{T(Psi)}
\end{alignat}
where $\xi_{k}^{ii}=$ $\left|\xi_{k}^{ii}\right\rangle \left\langle \xi_{k}^{ii}\right|$,
$\left|\xi_{k}^{ii}\right\rangle =\left|\alpha_{k}^{i}\right\rangle \left|\beta_{k}^{i}\right\rangle $;
$\xi_{k}^{ij^{\pm}}=\left|\xi_{k}^{ij^{\pm}}\right\rangle \left\langle \xi_{k}^{ij^{\pm}}\right|$,
$\left|\xi_{k}^{ij^{\pm}}\right\rangle =\frac{1}{\sqrt{2}}\left(\left|\alpha_{k}^{i}\right\rangle \left|\beta_{k}^{j}\right\rangle \pm\left|\alpha_{k}^{j}\right\rangle \left|\beta_{k}^{i}\right\rangle \right)$. 

Now decomposing $\mathbf{1}_{\mathcal{A}_{1}\otimes\mathcal{B}_{1}}$
as
\begin{alignat}{1}
\mathbf{1}_{\mathcal{A}_{1}\otimes\mathcal{B}_{1}} & =\sum_{i=1}^{d}\xi_{k}^{ii}+\sum_{i,j=1,\,i<j}^{d}\xi_{k}^{ij^{+}}+\sum_{i,j=1,\,i<j}^{d}\xi_{k}^{ij^{-}}\label{I}
\end{alignat}
and using \eqref{I} and \eqref{T(Psi)} in \eqref{upsilon-k}, we
get 
\begin{alignat}{1}
\Upsilon_{k} & =2\sum_{i,j=1,\,i<j}^{d}\xi_{k}^{ij^{-}},\label{Upsilon}
\end{alignat}
which is positive semidefinite. Similarly, from \eqref{I} and \eqref{Upsilon},
we get

\begin{alignat}{1}
\left(\mathbf{1}_{\mathcal{A}_{1}\otimes\mathcal{B}_{1}}-\frac{1}{2}\Upsilon_{k}\right) & =\sum_{i=1}^{d}\xi_{k}^{ii}+\sum_{i,j=1,\,i<j}^{d}\xi_{k}^{ij^{+}},\label{1-1/2=00005CUpsilon}
\end{alignat}
which is also positive semidefinite. 

The above analysis clearly holds for any $k\in\left\{ 1,\dots,d^{2}\right\} $,
and so we have proved \eqref{PositivityH} for all $k=1,\dots,d^{2}$.
This completes the proof of Lemma \ref{PPT-upper bound}. 

\section{\label{Lemma-lower-bound} Lower bound on the success probability:
Proof of Lemma \ref{LOCC-prob}}

The LOCC protocol for distinguishing the elements of $S_{\Psi\otimes\tau}$
is based on quantum teleportation. However, before we describe the
protocol, it will be helpful to understand how states transform under
teleportation. 

Recall that in the standard teleportation protocol of a $d$-dimensional
quantum state, the teleportation channel is taken to be the maximally
entangled state $\left|\Psi_{1}\right\rangle $ given by \eqref{Psi(1)},
Alice performs her measurement in the canonical maximally entangled
basis, communicates her outcome to Bob via a classical channel, and
Bob applies a unitary correction chosen from a fixed set of unitary
operators. Note that the unitary operator that Bob applies depends
on the outcome of Alice's measurement; in particular, there is an
one-to-one correspondence between Alice's outcome and Bob's unitary
correction. 

Suppose Alice wants to teleport a quantum state, say, $\left|\varphi\right\rangle \in\mathbb{C}^{d}$,
but instead of $\left|\Psi_{1}\right\rangle $ they share a maximally
entangled state $\left|\varPsi_{x}\right\rangle \in\mathcal{A}_{1}\otimes\mathcal{B}_{1}$,
where $\left|\varPsi_{x}\right\rangle =\left(\mathbf{1}_{\mathcal{A}_{1}}\otimes V_{x}\right)\left|\Psi_{1}\right\rangle $
for some unitary operator $V_{x}\in U\left(\mathcal{B}_{1}\right)$.
It is easy to check that if they carry out all the steps of the standard
protocol prescribed for $\left|\Psi_{1}\right\rangle $, Bob will
end up with the state $\left|\varphi_{x}\right\rangle =V_{x}\left|\varphi\right\rangle $.
Therefore, to reproduce the teleportation input correctly Bob needs
to apply another unitary correction $V_{x}^{-1}$ on $\left|\varphi_{x}\right\rangle $. 

Now suppose that Alice and Bob do not know the identity of $\left|\varPsi_{x}\right\rangle $
but only that it belongs to a known set, i.e., $\left|\varPsi_{x}\right\rangle \in\left\{ \left|\varPsi_{x_{1}}\right\rangle ,\dots,\left|\varPsi_{x_{n}}\right\rangle \right\} $,
where $\left|\varPsi_{x_{i}}\right\rangle =\left(\mathbf{1}_{\mathcal{A}_{1}}\otimes V_{x_{i}}\right)\left|\Psi_{1}\right\rangle $.
It then follows that $\left|\varphi_{x}\right\rangle \in\left\{ \left|\varphi_{x_{i}}\right\rangle ,\dots,\left|\varphi_{x_{n}}\right\rangle \right\} $,
where $\left|\varphi_{x_{i}}\right\rangle =V_{x_{i}}\left|\varphi\right\rangle $.
In this case, the input state cannot be exactly reproduced at Bob's
end because they do not know which maximally entangled state they
shared. 

This can be extended to teleportation of an entangled system. Suppose
the state Alice wants to teleport is $\left|\Theta\right\rangle \in\mathcal{A}^{\prime\prime}\otimes\mathcal{A}^{\prime}$,
where $\mathcal{A}^{\prime}=\mathcal{A}^{\prime\prime}=\mathbb{C}^{d}$.
Once again assume that the teleportation channel is $\left|\varPsi_{x}\right\rangle $.
Then at the end of the standard protocol, assuming Alice performs
her measurement on the composite system $\mathcal{A}^{\prime}\otimes\mathcal{A}_{1}$,
they will end up sharing $\left|\Theta_{x}\right\rangle \in\mathcal{A}^{\prime\prime}\otimes\mathcal{B}_{1}$,
where $\left|\Theta_{x}\right\rangle =\left(\mathbf{1}_{\mathcal{A}^{\prime\prime}}\otimes V_{x}\right)\left|\Theta\right\rangle $.
Moreover, if $\left|\varPsi_{x}\right\rangle \in\left\{ \left|\varPsi_{x_{1}}\right\rangle ,\dots,\left|\varPsi_{x_{n}}\right\rangle \right\} $,
then $\left|\Theta_{x}\right\rangle \in\left\{ \left|\Theta_{x_{i}}\right\rangle ,\dots,\left|\Theta_{x_{n}}\right\rangle \right\} $,
where $\left|\Theta_{x_{i}}\right\rangle =\left(\mathbf{1}_{\mathcal{A}^{\prime\prime}}\otimes V_{x_{i}}\right)\left|\Theta\right\rangle $. 

We now come to the LOCC protocol for distinguishing the elements of
$S_{\Psi\otimes\tau}$. The first step is teleportation of one half
of the resource state $\left|\tau\right\rangle $ using the unknown
state, say, $\left|\Psi_{i}\right\rangle \in\left\{ \left|\Psi_{1}\right\rangle ,\dots,\left|\Psi_{d^{2}}\right\rangle \right\} $
following the standard protocol. Noting that $\left|\Psi_{i}\right\rangle $
has the form \eqref{Psi(k)}, the initial state can be written as
\begin{alignat}{1}
\left|\Psi_{i}\right\rangle \otimes\left|\tau\right\rangle  & =\left(\mathbf{1}_{\mathcal{A}_{1}}\otimes U_{i}\right)\left|\Psi_{1}\right\rangle \otimes\left|\tau\right\rangle \in\mathcal{A}_{1}\otimes\mathcal{B}_{1}\otimes\mathcal{A}_{2}\otimes\mathcal{B}_{2}.\label{initial state}
\end{alignat}
Now they complete all the steps of the standard protocol (Alice performs
her measurement on $\mathcal{A}_{1}\otimes\mathcal{A}_{2}$, informs
Bob about the outcome, who applies the prescribed unitary correction
on $\mathcal{B}_{1}$). From our previous discussion, we know that
this results in Bob holding the state 
\begin{alignat}{1}
\left|\gamma_{i}\right\rangle  & =\left(\mathbf{1}_{\mathcal{B}_{2}}\otimes U_{i}\right)\left|\tau\right\rangle \in\mathcal{B}_{2}\otimes\mathcal{B}_{1}.\label{gamma-i}
\end{alignat}
Therefore, after teleportation, Bob holds one of $\left\{ \left|\gamma_{1}\right\rangle ,\dots,\left|\gamma_{d^{2}}\right\rangle \right\} \subset\mathcal{B}_{2}\otimes\mathcal{B}_{1}$.
The inner product between any pair of states is found to be 
\begin{alignat}{1}
\left\langle \gamma_{i}\vert\gamma_{j}\right\rangle  & =\sum_{k=1}^{d}a_{k}^{2}\left\langle k\left|U_{i}^{\dagger}U_{j}\right|k\right\rangle .\label{gammai-gamma-j}
\end{alignat}
As one would expect, the states $\left|\gamma_{i}\right\rangle $
are not all mutually orthogonal unless $\left|\tau\right\rangle $
is maximally entangled, in which case we have $a_{k}^{2}=1/d$ for
$k=1,\dots,d$, and $\left\langle \gamma_{i}\vert\gamma_{j}\right\rangle =\frac{1}{d}\text{Tr}\left(U_{i}^{\dagger}U_{j}\right)=\delta_{ij}$. 

Now observe that the elements of $\left\{ \left|\Psi_{1}\right\rangle ,\dots,\left|\Psi_{d^{2}}\right\rangle \right\} $
are in one-to-one correspondence with that of $\left\{ \left|\gamma_{1}\right\rangle ,\dots,\left|\gamma_{d^{2}}\right\rangle \right\} $,
i.e., if the unknown state was $\left|\Psi_{i}\right\rangle $, Bob's
state is guaranteed to be $\left|\gamma_{i}\right\rangle $. Therefore,
the ``teleportation'' step of the protocol induces the map $\left|\Psi_{i}\right\rangle \rightarrow\left|\gamma_{i}\right\rangle $
for each $i=1,\dots,d^{2}$, and consequently, the problem is mapped
onto that of distinguishing between the elements of $\left\{ \left|\gamma_{1}\right\rangle ,\dots,\left|\gamma_{d^{2}}\right\rangle \right\} $,
each given with probability $1/d^{2}$ (since the unknown state $\left|\Psi_{i}\right\rangle $
was selected with probability $1/d^{2}$). 

To distinguish between the states $\left|\gamma_{i}\right\rangle $
Bob performs a measurement in the orthonormal, maximally entangled
basis $\left\{ \left|\Psi_{1}\right\rangle ,\dots,\left|\Psi_{d^{2}}\right\rangle \right\} \subset\mathcal{B}_{2}\otimes\mathcal{B}_{1}$,
where the basis vectors are defined as $\left|\Psi_{i}\right\rangle =\left(\mathbf{1}_{\mathcal{B}_{2}}\otimes U_{i}\right)\left|\Psi_{1}\right\rangle $.
The outcome of this measurement is now concluded as the unknown state
that they initially shared. The success probability for this is given
by 
\begin{alignat}{1}
p & =\frac{1}{d^{2}}\sum_{i=1}^{d^{2}}\left|\left\langle \Psi_{i}\vert\gamma_{i}\right\rangle \right|^{2}\nonumber \\
 & =\frac{1}{d^{2}}\sum_{i=1}^{d^{2}}\left|\left\langle \Psi_{1}\vert\left(\mathbf{1}_{\mathcal{B}_{2}}\otimes U_{i}^{\dagger}\right)\left(\mathbf{1}_{\mathcal{B}_{2}}\otimes U_{i}\right)\vert\tau\right\rangle \right|^{2}\nonumber \\
 & =\frac{1}{d^{2}}\sum_{i=1}^{d^{2}}\left|\left\langle \Psi_{1}\vert\tau\right\rangle \right|^{2}\nonumber \\
 & =\left|\left\langle \Psi_{1}\vert\tau\right\rangle \right|^{2}=\frac{1}{d}\left(\sum_{i=1}^{d}a_{i}\right)^{2}=F\left(\tau\right).\label{p(L)-lower-bound}
\end{alignat}
This completes the proof. 

\section{\label{incomplete} Distinguishing the elements of an incomplete
orthonormal maximally entangled basis }

Let us now suppose that the we are given a set of $N$ orthonormal
maximally entangled states in $\mathcal{A}_{1}\otimes\mathcal{B}_{1}$
given by 
\begin{alignat}{1}
\mathcal{S}_{\Psi} & =\left\{ \left(\frac{1}{N},\left|\Psi_{k}\right\rangle \right):k=1,\dots,N\right\} ,\label{S-prime-Psi}
\end{alignat}
where $N\in\left[d+1,d^{2}\right]$, and the states are equally likely.
Such a set is locally indistinguishable \citep{Nathanson-2005}. Note
that sets of $N\leqslant d$ maximally entangled states may or may
not be locally indistinguishable, so we will leave them out from our
discussion. 

Once again the resource state is taken to be $\left|\tau\right\rangle \in\mathcal{A}_{2}\otimes\mathcal{B}_{2}$
given by \eqref{tau}. Thus the set of interest is
\begin{alignat}{1}
\mathcal{S}_{\Psi\otimes\tau} & =\left\{ \left(\frac{1}{N},\left|\Psi_{k}\right\rangle \otimes\left|\tau\right\rangle \right):k=1,\dots,N\right\} ,\label{Psi=00005Cotimestau}
\end{alignat}
where $N\in\left[d+1,d^{2}\right]$. In this case, we have been able
to obtain upper and lower bounds on the local success probability. 

Proceeding along the lines of $N=d^{2}$ with suitable modification
of $\mathbb{H}$ given by \eqref{H}, one finds that
\begin{alignat}{1}
p_{_{{\rm L}}}\left(\mathcal{S}_{\Psi\otimes\tau}\right) & \leqslant p_{_{{\rm PPT}}}\left(\mathcal{S}_{\Psi\otimes\tau}\right)\leqslant\frac{d^{2}}{N}F\left(\tau\right);N\in\left[d+1,d^{2}\right].\label{eqTheorem-1-1}
\end{alignat}
Note that for $N=d^{2}$, we get back our familiar bound derived in
the previous section. Further note that the upper bound is nontrivial
whenever $F\left(\tau\right)\leqslant N/d^{2}$. In particular, it
follows that if the elements of $\mathcal{S}_{\Psi}$ can be perfectly
distinguished using LOCC and a resource state $\left|\tau\right\rangle $,
then $F\left(\tau\right)\geqslant N/d^{2}$. Also observe that if
$\left|\tau\right\rangle $ is a product state (that is, no entanglement),
its fully entangled fraction is given by $1/d$, in which case the
upper bound reduces to $d/N$ reproducing the result of Nathanson
\citep{Nathanson-2005}. 

To obtain a lower bound we once again take recourse to the teleportation
protocol. The first step is identical and after its completion Bob
now needs to distinguish between the elements of $\left\{ \left|\gamma_{1}\right\rangle ,\dots,\left|\gamma_{N}\right\rangle \right\} $,
where $N\in\left[d+1,d^{2}\right]$. Earlier, for $N=d^{2}$, Bob
performed a measurement in the orthonormal, maximally entangled basis
$\left\{ \left|\Psi_{1}\right\rangle ,\dots,\left|\Psi_{d^{2}}\right\rangle \right\} \subset\mathcal{B}_{2}\otimes\mathcal{B}_{1}$,
where $\left|\Psi_{i}\right\rangle =\left(\mathbf{1}_{\mathcal{B}_{2}}\otimes U_{i}\right)\left|\Psi_{1}\right\rangle $,
and the outcome of this measurement was concluded as the unknown state
they initially shared. 

Now the situation is different, so there are several possibilities.
Here we discuss one such possibility and briefly sketch another. Bob
performs a measurement in an orthonormal basis $\left\{ \left|\Psi_{1}\right\rangle ,\dots,\left|\Psi_{N}\right\rangle ,\left|\psi_{1}^{\prime}\right\rangle ,\dots,\left|\psi_{d^{2}-N}^{\prime}\right\rangle \right\} \subset\mathcal{B}_{2}\otimes\mathcal{B}_{1}$,
where $\left|\Psi_{i}\right\rangle =\left(\mathbf{1}_{\mathcal{B}_{2}}\otimes U_{i}\right)\left|\Psi_{1}\right\rangle $
for $i=1,\dots,N$, and $\left|\psi_{1}^{\prime}\right\rangle ,\dots,\left|\psi_{d^{2}-N}^{\prime}\right\rangle $
are orthonormal states that may or may not be maximally entangled.
For outcomes $1,\dots,N$, the strategy would be the same - for the
$i$-th outcome, the unknown state is concluded as $\left|\Psi_{i}\right\rangle $.
The question is what should the mapping be when the outcome is one
of the prime states. Here one could adopt the following strategy:
Given such an outcome, one first identifies which state amongst $\left|\gamma_{1}\right\rangle ,\dots,\left|\gamma_{N}\right\rangle $
is most likely to give this outcome: if this state is $\left|\gamma_{x}\right\rangle $,
the unknown state is concluded as $\left|\Psi_{x}\right\rangle $
for $x\in\left[1,N\right]$. With this, one can compute the success
probability as
\begin{alignat}{1}
p & =\frac{1}{N}\sum_{i=1}^{N}\left|\left\langle \Psi_{i}\vert\gamma_{i}\right\rangle \right|^{2}+\frac{1}{N}\sum_{i=1}^{d^{2}-N}\max_{j}\left|\left\langle \psi_{i}^{\prime}\vert\gamma_{j}\right\rangle \right|^{2}\nonumber \\
 & =F\left(\tau\right)+\frac{1}{N}\sum_{i=1}^{d^{2}-N}\max_{j}\left|\left\langle \psi_{i}^{\prime}\vert\gamma_{j}\right\rangle \right|^{2},\label{p-lowe}
\end{alignat}
which is a lower bound on $p_{_{{\rm L}}}\left(\mathcal{S}_{\Psi\otimes\tau}\right)$.
The second term on the right-hand side can be exactly computed. Note
that we have complete freedom to choose the states $\left|\psi_{1}^{\prime}\right\rangle ,\dots,\left|\psi_{d^{2}-N}^{\prime}\right\rangle $
and the choice should be such that the second term is maximized. 

Alternatively, one can go for a simpler measurement with $N+1$ outcomes,
where the $\left(N+1\right)$-th operator is a projection operator
onto the subspace orthogonal to the subspace spanned by $\left|\Psi_{1}\right\rangle ,\dots,\left|\Psi_{N}\right\rangle $.
In this case, whenever this outcome clicks we simply pick the state
$\left|\gamma_{x}\right\rangle $ that is most likely to produce this
outcome and conclude about the unknown state accordingly. 

Combining \eqref{eqTheorem-1-1} and \eqref{p-lowe}, we therefore
have
\begin{eqnarray}
F\left(\tau\right)+\frac{1}{N}\sum_{i=1}^{d^{2}-N}\max_{j}\left|\left\langle \psi_{i}^{\prime}\vert\gamma_{j}\right\rangle \right|^{2}\leqslant & p_{_{{\rm L}}}\left(\mathcal{S}_{\Psi\otimes\tau}\right) & \leqslant\frac{d^{2}}{N}F\left(\tau\right);\;N\in\left[d+1,d^{2}\right].\label{lower and upper}
\end{eqnarray}
This concludes our discussion on distinguishing the elements of an
incomplete maximally entangled basis using a partially entangled state
as a resource. 

\section{\label{Conclusions} Conclusions }

One of the central questions in local distinguishability of quantum
states is how well a given set of locally indistinguishable states
can be distinguished using LOCC and shared entanglement as a resource.
In this paper, we considered the problem of locally distinguishing
the elements of a bipartite maximally entangled orthonormal basis
using a partially entangled state acting as a resource. In particular,
our objective was to find an expression for the success probability,
which quantifies how well the basis states can be distinguished in
an entanglement-assisted LOCC setup. This has been solved for the
Bell basis in $\mathbb{C}^{2}\otimes\mathbb{C}^{2}$ \citep{B-IQC-2015}
but remained open in all higher dimensions. Here, we solved this problem
for any maximally entangled orthonormal basis in dimensions $\mathbb{C}^{d}\otimes\mathbb{C}^{d}$,
where $d\geqslant3$. Assuming the basis states are uniformly distributed,
we derived an exact formula for the success probability. This formula
corresponds to the fully entangled fraction of the resource state.
So, how well the elements of a maximally entangled basis can be locally
distinguished using shared entanglement is determined by how close
the resource state is to a maximally entangled state. 

To derive our formula, we proceeded as follows. First, we established
an upper bound on the success probability using the SDP formulation
of distinguishing quantum states by PPT measurements \citep{Cosentino-2013}
and then showed that this upper bound is achievable by LOCC. The LOCC
protocol here is based on quantum teleportation. One further concludes
that, for distinguishing maximally entangled states with shared entanglement,
PPT measurements provide no advantage over LOCC, even though LOCC
is a strict subset of PPT measurements.

We now briefly mention a couple of open problems. Let $S\subset\mathbb{C}^{d}\otimes\mathbb{C}^{d}$
for $d\geqslant2$ be an orthonormal set of maximally entangled states.
Such a set is locally indistinguishable whenever $d+1\leqslant\left|S\right|\leqslant d^{2}$
\citep{Nathanson-2005}. In addition, locally indistinguishable sets
with $\left|S\right|\leqslant d$ also exist for $d\geqslant4$ \citep{Cosentino-2013,Cosentino-Russo-2014}.
A fundamental question is whether a maximally entangled state is always
necessary to perfectly distinguish an orthonormal set of maximally
entangled states known to be locally indistinguishable. This question
is completely solved for $d=2$ \citep{B-IQC-2015} but remains open
in higher dimensions except when $S$ is a basis of $\mathbb{C}^{d}\otimes\mathbb{C}^{d}$.
One could attempt to answer this question in higher dimensions by
computing the entanglement cost or the success probability within
the framework of entanglement-assisted LOCC as discussed in this paper
or earlier works \citep{B-IQC-2015,BGKW-2009,BRW-2010,BHN-2018,BR-2021}.
The latter approach, which involves computing the success probability,
as we did in this paper for $\left|S\right|=d^{2}$, could be fruitful
as one could directly check whether the resource state needs to be
maximally entangled or not by setting the success probability to unity.
In this paper, we made some progress by presenting upper and lower
bounds on the success probability, both of which can be exactly computed.
\\
\\

\begin{acknowledgement*}
We thank Tathagata Gupta and Shayeef Murshid for helpful comments. 
\end{acknowledgement*}

\section*{Appendix: Proof of Lemma \ref{H and Hprime}}

To prove Lemma \ref{H and Hprime} we will use the following result.
Let ${\rm Lin}\left(\mathcal{H}\otimes\mathcal{H}^{\prime}\right)$
denote the set of linear operators on $\mathcal{H}\otimes\mathcal{H}^{\prime}$,
where $\mathcal{H}$ and $\mathcal{H}^{\prime}$ are finite-dimensional
state spaces. 
\begin{lem}
\label{swap-T=00003DT-swap} Let $\mathcal{X}_{1}=\mathcal{X}_{2}=\mathcal{Y}_{1}=\mathcal{Y}_{2}=\mathbb{C}^{d}$,
where $d\geqslant2$. For any pair of linear operators $\Lambda\in{\rm Lin}\left(\mathcal{X}_{1}\otimes\mathcal{Y}_{1}\right)$
and $\Xi\in{\rm Lin}\left(\mathcal{X}_{2}\otimes\mathcal{Y}_{2}\right)$,
it holds that 
\begin{alignat}{1}
U_{\mathcal{Y}_{1}\leftrightarrow\mathcal{X}_{2}}\left[{\rm \left(T_{\mathcal{X}_{1}}\otimes{\rm T}_{\mathcal{X}_{2}}\right)}\left(\Lambda\otimes\Xi\right)\right]U_{\mathcal{Y}_{1}\leftrightarrow\mathcal{X}_{2}}^{\dagger} & ={\rm T}_{\mathcal{X}}\left[U_{\mathcal{Y}_{1}\leftrightarrow\mathcal{X}_{2}}\left(\Lambda\otimes\Xi\right)U_{\mathcal{Y}_{1}\leftrightarrow\mathcal{X}_{2}}^{\dagger}\right],\label{identity-1-1}
\end{alignat}
where $U_{\mathcal{Y}_{1}\leftrightarrow\mathcal{X}_{2}}$ is the
unitary swap operator defined similarly as in \eqref{eU-swap} and
$\mathcal{X}=\mathcal{X}_{1}\otimes\mathcal{X}_{2}$. 
\end{lem}
\begin{proof}
Let $\left\{ \left|x_{i}\right\rangle :i=1,\dots,d\right\} $ be the
standard basis of $\mathcal{X}_{1}$ and $\mathcal{X}_{2}$ and $\left\{ \left|y_{i}\right\rangle :i=1,\dots,d\right\} $
be the standard bases of $\mathcal{Y}_{1}$ and $\mathcal{Y}_{2}$.
Linear operators $\Lambda$ acting on $\mathcal{X}_{1}\otimes\mathcal{Y}_{1}$
and $\Xi$ on $\mathcal{X}_{2}\otimes\mathcal{Y}_{2}$ can be written
as
\begin{alignat}{1}
\Lambda & =\sum_{a,b,c,e}q_{ab}^{ce}\left|x_{a}\right\rangle \left\langle x_{b}\right|\otimes\left|y_{c}\right\rangle \left\langle y_{e}\right|,\label{|Lambda}\\
\Xi & =\sum_{\alpha,\beta,\gamma,\delta}q_{\alpha\beta}^{\gamma\delta}\left|x_{\alpha}\right\rangle \left\langle x_{\beta}\right|\otimes\left|y_{\gamma}\right\rangle \left\langle y_{\delta}\right|,\label{Xi}
\end{alignat}
where each index runs from $1$ to $d$. 

Therefore, 
\begin{alignat}{1}
\Lambda\otimes\Xi & =\sum_{\bm{i}}q_{ab}^{ce}q_{\alpha\beta}^{\gamma\delta}\left|x_{a}\right\rangle \left\langle x_{b}\right|\otimes\left|y_{c}\right\rangle \left\langle y_{e}\right|\otimes\left|x_{\alpha}\right\rangle \left\langle x_{\beta}\right|\otimes\left|y_{\gamma}\right\rangle \left\langle y_{\delta}\right|,\label{lambdaXxi}
\end{alignat}
where $\bm{i}=\left\{ a,b,c,e,\alpha,\beta,\gamma,\delta\right\} $
is our notation for the collection of indices and will be used in
the rest of the proof. 

Let us now compute the LHS of \eqref{identity-1-1}. First, 
\begin{alignat}{1}
{\rm \left(T_{\mathcal{X}_{1}}\otimes{\rm T}_{\mathcal{X}_{2}}\right)}\left(\Lambda\otimes\Xi\right) & ={\rm T_{\mathcal{X}_{1}}}\left(\Lambda\right)\otimes{\rm {\rm T}_{\mathcal{X}_{2}}}\left(\Xi\right),\nonumber \\
 & =\sum_{\bm{i}}q_{ab}^{ce}q_{\alpha\beta}^{\gamma\delta}\left|x_{b}\right\rangle \left\langle x_{a}\right|\otimes\left|y_{c}\right\rangle \left\langle y_{e}\right|\otimes\left|x_{\beta}\right\rangle \left\langle x_{\alpha}\right|\otimes\left|y_{\gamma}\right\rangle \left\langle y_{\delta}\right|.\label{T(lambdaXXi)}
\end{alignat}
Next, swapping $\mathcal{Y}_{1}$ and $\mathcal{X}_{2}$ gives
\begin{alignat}{1}
\text{LHS of \eqref{identity-1-1}} & =\sum_{\bm{i}}q_{ab}^{ce}q_{\alpha\beta}^{\gamma\delta}\left|x_{b}\right\rangle \left\langle x_{a}\right|\otimes\left|x_{\beta}\right\rangle \left\langle x_{\alpha}\right|\otimes\left|y_{c}\right\rangle \left\langle y_{e}\right|\otimes\left|y_{\gamma}\right\rangle \left\langle y_{\delta}\right|,\nonumber \\
 & =\sum_{\bm{i}}q_{ab}^{ce}q_{\alpha\beta}^{\gamma\delta}\left|x_{b}x_{\beta}\right\rangle \left\langle x_{a}x_{\alpha}\right|\otimes\left|y_{c}y_{\gamma}\right\rangle \left\langle y_{e}y_{\delta}\right|.\label{swap-T}
\end{alignat}
We now compute the RHS of \eqref{identity-1-1}. 
\begin{alignat}{1}
{\rm T}_{\mathcal{X}}\left[U_{\mathcal{Y}_{1}\leftrightarrow\mathcal{X}_{2}}\left(\Lambda\otimes\Xi\right)U_{\mathcal{Y}_{1}\leftrightarrow\mathcal{X}_{2}}^{\dagger}\right] & ={\rm T}_{\mathcal{X}}\left(\sum_{\bm{i}}q_{ab}^{ce}q_{\alpha\beta}^{\gamma\delta}\left|x_{a}\right\rangle \left\langle x_{b}\right|\otimes\left|x_{\alpha}\right\rangle \left\langle x_{\beta}\right|\otimes\left|y_{c}\right\rangle \left\langle y_{e}\right|\otimes\left|y_{\gamma}\right\rangle \left\langle y_{\delta}\right|\right),\nonumber \\
 & ={\rm T}_{\mathcal{X}}\left(\sum_{\bm{i}}q_{ab}^{ce}q_{\alpha\beta}^{\gamma\delta}\left|x_{a}x_{\alpha}\right\rangle \left\langle x_{b}x_{\beta}\right|\otimes\left|y_{c}y_{\gamma}\right\rangle \left\langle y_{e}y_{\delta}\right|\right),\nonumber \\
 & =\sum_{\bm{i}}q_{ab}^{ce}q_{\alpha\beta}^{\gamma\delta}\left|x_{b}x_{\beta}\right\rangle \left\langle x_{a}x_{\alpha}\right|\otimes\left|y_{c}y_{\gamma}\right\rangle \left\langle y_{e}y_{\delta}\right|.\label{T-swap}
\end{alignat}
From \eqref{swap-T} and \eqref{T-swap}, we see that
\begin{alignat}{1}
\text{LHS of \eqref{identity-1-1}} & =\text{RHS of \eqref{identity-1-1}}.\label{LHS=00003DRHS}
\end{alignat}
\end{proof}

\subsubsection*{Proof of Lemma \ref{H and Hprime}}

Suppose that 
\begin{alignat}{1}
{\rm \left(T_{\mathcal{A}_{1}}\otimes{\rm T}_{\mathcal{A}_{2}}\right)}\left(\mathbb{H}-\frac{1}{d^{2}}\Psi_{k}\otimes\tau\right) & \in{\rm Pos}\left(\mathcal{A}_{1}\otimes\mathcal{B}_{1}\otimes\mathcal{A}_{2}\otimes\mathcal{B}_{2}\right),\;k=1,\dots,d^{2},\label{positive-condition}
\end{alignat}
where $\mathbb{H}$ is a Hermitian operator of the form given by \eqref{Hprime-1}. 

Since swapping subsystems does not affect the positivity of an operator,
it holds that
\begin{alignat}{1}
U_{\mathcal{B}_{1}\leftrightarrow\mathcal{A}_{2}}\left[{\rm \left(T_{\mathcal{A}_{1}}\otimes{\rm T}_{\mathcal{A}_{2}}\right)}\left(\mathbb{H}-\frac{1}{d^{2}}\Psi_{k}\otimes\tau\right)\right]U_{\mathcal{B}_{1}\leftrightarrow\mathcal{A}_{2}}^{\dagger} & \in{\rm Pos}\left(\mathcal{A}\otimes\mathcal{B}\right),\;k=1,\dots,d^{2}.\label{SWAP(T1XT2)-positive-1}
\end{alignat}
The operator in \eqref{SWAP(T1XT2)-positive-1} can be expanded as
\begin{multline}
U_{\mathcal{B}_{1}\leftrightarrow\mathcal{A}_{2}}\left[{\rm \left(T_{\mathcal{A}_{1}}\otimes{\rm T}_{\mathcal{A}_{2}}\right)}\left(\mathbb{H}-\frac{1}{d^{2}}\Psi_{k}\otimes\tau\right)\right]U_{\mathcal{B}_{1}\leftrightarrow\mathcal{A}_{2}}^{\dagger}\\
=U_{\mathcal{B}_{1}\leftrightarrow\mathcal{A}_{2}}\left[{\rm \left(T_{\mathcal{A}_{1}}\otimes{\rm T}_{\mathcal{A}_{2}}\right)}\mathbb{H}\right]U_{\mathcal{B}_{1}\leftrightarrow\mathcal{A}_{2}}^{\dagger}-\frac{1}{d^{2}}U_{\mathcal{B}_{1}\leftrightarrow\mathcal{A}_{2}}\left[{\rm \left(T_{\mathcal{A}_{1}}\otimes{\rm T}_{\mathcal{A}_{2}}\right)}\left(\Psi_{k}\otimes\tau\right)\right]U_{\mathcal{B}_{1}\leftrightarrow\mathcal{A}_{2}}^{\dagger}.\label{expanded}
\end{multline}
We now apply Lemma \ref{swap-T=00003DT-swap} to each term on the
RHS of \eqref{expanded}. For the first term, we get 
\begin{alignat}{1}
U_{\mathcal{B}_{1}\leftrightarrow\mathcal{A}_{2}}\left[{\rm \left(T_{\mathcal{A}_{1}}\otimes{\rm T}_{\mathcal{A}_{2}}\right)}\mathbb{H}\right]U_{\mathcal{B}_{1}\leftrightarrow\mathcal{A}_{2}}^{\dagger} & ={\rm T}_{\mathcal{A}}\left(U_{\mathcal{B}_{1}\leftrightarrow\mathcal{A}_{2}}\mathbb{H}U_{\mathcal{B}_{1}\leftrightarrow\mathcal{A}_{2}}^{\dagger}\right),\label{identity-1}
\end{alignat}
and for the second 
\begin{alignat}{1}
\frac{1}{d^{2}}U_{\mathcal{B}_{1}\leftrightarrow\mathcal{A}_{2}}\left[{\rm \left(T_{\mathcal{A}_{1}}\otimes{\rm T}_{\mathcal{A}_{2}}\right)}\left(\Psi_{k}\otimes\tau\right)\right]U_{\mathcal{B}_{1}\leftrightarrow\mathcal{A}_{2}}^{\dagger} & =\frac{1}{d^{2}}{\rm T}_{\mathcal{A}}\left(\Phi_{k}\right),\;k=1,\dots,d^{2},\label{identity-2}
\end{alignat}
where to arrive at \eqref{identity-2} we have used $\Phi_{k}=U_{\mathcal{B}_{1}\leftrightarrow\mathcal{A}_{2}}\left(\Psi_{k}\otimes\tau\right)U_{\mathcal{B}_{1}\leftrightarrow\mathcal{A}_{2}}^{\dagger}$.

Now using \eqref{identity-1} and \eqref{identity-2}, we can write
\eqref{expanded} as
\begin{alignat}{1}
U_{\mathcal{B}_{1}\leftrightarrow\mathcal{A}_{2}}\left[{\rm \left(T_{\mathcal{A}_{1}}\otimes{\rm T}_{\mathcal{A}_{2}}\right)}\left(\mathbb{H}-\frac{1}{d^{2}}\Psi_{k}\otimes\tau\right)\right]U_{\mathcal{B}_{1}\leftrightarrow\mathcal{A}_{2}}^{\dagger} & ={\rm T}_{\mathcal{A}}\left(U_{\mathcal{B}_{1}\leftrightarrow\mathcal{A}_{2}}\mathbb{H}U_{\mathcal{B}_{1}\leftrightarrow\mathcal{A}_{2}}^{\dagger}-\Phi_{k}/d^{2}\right).\label{SWAP-t=00003DT-SWAP}
\end{alignat}
Since the LHS of \eqref{SWAP-t=00003DT-SWAP} is positive semidefinite
{[}see \eqref{SWAP(T1XT2)-positive-1}{]}, the RHS must also be so,
i.e., 
\begin{alignat}{1}
{\rm T}_{\mathcal{A}}\left(U_{\mathcal{B}_{1}\leftrightarrow\mathcal{A}_{2}}\mathbb{H}U_{\mathcal{B}_{1}\leftrightarrow\mathcal{A}_{2}}^{\dagger}-\Phi_{k}/d^{2}\right) & \in{\rm Pos}\left(\mathcal{A}\otimes\mathcal{B}\right),\;k=1,\dots,d^{2}.\label{final-condition-2}
\end{alignat}
Defining $H$ as 
\begin{alignat}{1}
H & =U_{\mathcal{B}_{1}\leftrightarrow\mathcal{A}_{2}}\mathbb{H}U_{\mathcal{B}_{1}\leftrightarrow\mathcal{A}_{2}}^{\dagger}\in{\rm Herm}\left(\mathcal{A}\otimes\mathcal{B}\right),\label{H and H}
\end{alignat}
\eqref{final-condition-2} takes the form 
\begin{alignat}{1}
{\rm T}_{\mathcal{A}}\left(H-\Phi_{k}/d^{2}\right) & \in{\rm Pos}\left(\mathcal{A}\otimes\mathcal{B}\right),\;k=1,\dots,d^{2}.\label{final-condition-1}
\end{alignat}
Thus the feasibility condition of the dual problem is satisfied for
an $H$ defined in \eqref{H and H}, provided the operator ${\rm \left(T_{\mathcal{A}_{1}}\otimes{\rm T}_{\mathcal{A}_{2}}\right)}\left(\mathbb{H}-\frac{1}{d^{2}}\Psi_{k}\otimes\tau\right)$
is positive semidefinite. This completes the proof.

\end{document}